\documentclass[preprint,12pt]{elsarticle}

\usepackage{amssymb}
\usepackage{amsthm}

\usepackage{amsmath}

\newtheorem{thm}{Theorem}[section]
\newtheorem{lem}{Lemma}[section]
\newtheorem{prop}{Proposition}[section]
\newtheorem{cor}{Corollary}[section]

\def\choose#1#2{ \begin{pmatrix} #1 \\ #2 \end{pmatrix} }

\journal{Journal of Mathematical Analysis and Applications}

\begin{document}

\begin{frontmatter}



\title{Approximating Annual Mean Incoming Solar Radiation}


\author[label1,label2]{Alice Nadeau}
\author[label2]{Richard McGehee} 

\address[label1]{583 Malott Hall\\ 
Department of Mathematics\\
Cornell University\\
Ithaca, NY 14850\\
 {\tt a.nadeau@cornell.edu}}
\address[label2]{School of Mathematics\\
University of Minnesota }

\begin{abstract}
We derive the Legendre series expansion for the insolation distribution on rapidly rotating planets as a function of sine of the latitude and the planet's obliquity. We give an explicit formula for the coefficients of this series as it depends on the obliquity and approximate the convergence rate.  We determine the optimal truncation of the series for use in climate models and compare this approximation to other approximations in the literature.
\end{abstract}


\begin{highlights}
\item Explicit formula for coefficients of series approximation for insolation is given
\item Optimal truncation of the series for use in climate models is computed
\item Optimal truncation is best continuous approximation in the literature
\end{highlights}


\end{frontmatter}


\section{Introduction}

Interest in modeling the climates of other planets has recently been ignited due to the fly-by of the Pluto-Charon system by the NASA probe New Horizons \cite{Stern2018} and the discovery of seven Earth-sized planets orbiting the nearby star TRAPPIST-1 and over 4,000 other planets outside our solar system \cite{Gillon2017,Thompson2018}.  
Surprised by the complexity of Pluto and Charon and inspired by the prospect of liquid water and life in the TRAPPIST-1 system, scientists are now trying to understand these observations through the use of mathematical models \cite{Earle2017a,Earle2018,Checlair2017}.  

An important component of climate models of any complexity is the amount of sunlight reaching the planet's surface, referred to as \emph{insolation} (for incoming solar radiation). Although approximations to planetary insolation distribution already exist in the literature many do not have explicit dependence on the planet's orbital parameters or are not continuous functions across the planet's surface, causing problems in the modeling framework (e.g. discussion in \cite{Earle2017a}).   Here we present an approximation that has explicit dependence on the planet's orbital parameters, is continuous across the planet's surface, and has lower error to the true insolation distribution than other approximations in the literature.

The insolation at any point on a planet is a function of the latitude and longitude of the point, the planet's orbital parameters (semi-major axis, obliquity, and precession angle), the position of the planet along its orbit, and the solar energy output.  Using Kepler's laws and integrating over an entire year, one can show that the global annual average power flux (Watts per square meter) is given by
\[
  Q(a,e) = \frac{K}{a^2\sqrt{1 - e^2}},
\]
where K is proportional to the solar output, \(a\) is the semi-major axis of the planet's orbit, and \(e\) is the eccentricity \cite{McGehee2012}.  This energy is then distributed across the planet's surface and is dependent on the planet's rotation rate relative to its orbital rate, obliquity angle, and precession angle. 

Here we restrict our study to the insolation distribution on planets that are rapidly rotating. While it should be noted that there is no definitive definition of ``rapid rotation'' in the literature, it is generally agreed that Earth is a rapidly rotating planet. The only planets in our solar system with slower rotation rates than Earth are Mercury (with 3 rotations every 2 revolutions) and Venus (with $-0.92$ rotations every revolution).  Mercury and Venus are both slowly rotating by the colloquial definition used in the literature (e.g. \cite{Dobrovolskis2009,Dobrovolskis2013,Dobrovolskis2015}).  
We will define rapid rotation to be any rotation rate which causes the annual average insolation distribution to be rotationally symmetric about the planet's axis of rotation.

For a rapidly rotating planet, the orbital parameters and the position of the planet do not change substantially during a day, leading to a simplification of distribution by latitude of the annual average insolation. In this case, annual average insolation distribution reduces to a function only of the obliquity ($\beta$) and latitude ($\overline\phi$) \cite{Ward1974, McGehee2012, Dobrovolskis2013} and is given by
\begin{equation} \label{eq:s}
  s(\overline\phi,\beta) = \frac{2}{\pi^2}
    \int_0^{2\pi} \sqrt{ 1 - 
      \left(
        \cos\overline\phi\sin\beta\sin\gamma - \sin\overline\phi\cos\beta
      \right)^2 } d\gamma ,
\end{equation}
where \(\gamma\) is the longitude.  Since the latitude is measured up and down from the equator, we have \(-\pi/2 \le \overline\phi \le \pi/2\), while, since obliquity is the angle between the angular momentum vector of the planetary orbit and the angular momentum vector of the planetary spin, we have \(0 \le \beta \le \pi\).

For each fixed obliquity \(\beta\), \(s(\overline\phi,\beta)\) is the distribution of insolation across the surface of the planet, so the annual average insolation at latitude \(\overline\phi\) is given by
\[
  Q(a,e)s(\phi,\beta).
\]
We derive an infinite series representation of the function $s(\overline\phi,\beta)$ in terms of Legendre polynomials (Theorem \ref{Thm-Main}).  Truncating this series gives a polynomial approximation for the insolation function, allowing for faster computation of the insolation while also avoiding the numerical approximation of the integral.  
A quadratic approximation of $s$ for the Earth's obliquity has been used  extensively (e.g. \cite{Abbot2011,Barry2017,McGehee2012,McGehee2014,Walsh2015,Widiasih2013}).  However, for other planets, a quadratic approximation fails to capture the qualitative behavior of the insolation as a function of latitude.  In a previous paper \cite{Nadeau2017}, Nadeau and McGehee introduced a sixth order polynomial approximation and showed that it captures the characteristics of Pluto's insolation.  Here that result is generalized and placed on a firm mathematical foundation using classical results about spherical harmonics.

In modeling studies, it is usually most appropriate to take sine of the latitude instead of latitude so that the infinitesimal $d\eta=\cos\overline\phi d\overline\phi$ is proportional to the area of the latitudinal strip parallel to $\overline\phi$.  Taking cosine of obliquity makes $s$ symmetric in sine of the latitude ($\eta$) and cosine of obliquity ($\zeta$):
\begin{align}
  s(\eta,\zeta) =  \frac{2}{\pi^2}
    \int_0^{2\pi} \sqrt{ 1 - 
      \left(
        \sqrt{1-\eta^2}\sqrt{1-\zeta^2}\sin\gamma - \eta\zeta
      \right)^2 }\,d\gamma .
\label{EQ-insol-fast-rotation}
\end{align}

The Legendre polynomials \(P_i(x), i = 0,1,2,\ldots ,\) form a complete orthogonal set in the space \(L^2([-1,1])\) with the properties \(P_i\) has degree \(i\) and \(P_i(1)=1\) \cite{BatemanII}.  Therefore, the products 
\(P_{i,j}(x,y) = P_i(x)P_j(y)\) form a complete orthogonal set in the space
\(L^2([-1,1]\times[-1,1])\).  Thus we can write
\begin{equation} \label{eq:sfullexpansion}
  s(\eta,\zeta) = 
  \sum_{i=0}^\infty \sum_{j=0}^\infty c_{ij}P_i(\zeta)P_j(\eta).
\end{equation}
The series naturally converges in \(L^2\), and the convergence is also pointwise (see Section \ref{Section-Conv}).   Surprisingly, \(c_{ij}\) is diagonal, in particular:
\begin{thm} \label{Thm-Main}
The annual average insolation distribution function \eqref{EQ-insol-fast-rotation} can be written
\begin{align}
  s(\eta,\zeta) = \sum_{n=0}^\infty A_{2n}P_{2n}(\zeta)P_{2n}(\eta) ,
\label{EQ:Insol-Approx}
\end{align}
where \(P_{2n}\) is the Legendre polynomial of degree \(2n\), and where
\[  
  A_{2n} = 
    \frac{(-1)^n(4n+1)}{2^{2n-1}}
    \sum_{k=0}^n \choose{2n}{n-k} \choose{2n+2k}{2k} \choose{1/2}{k+1}.
\]
\end{thm}
\noindent Here we are using the standard notation
\[
  \choose{r}{j} = \frac{r(r-1)\cdots(r-j+1)}{j!}.
\]
\noindent The proof of this theorem relies on two main lemmas which are stated and proved in the following two sections.  The proof of the theorem is given in Section \ref{Proof-Main-Thm}.  In Section \ref{Section-Conv} we discuss convergence properties of the approximation and in Section \ref{Section-Comp} we compare the approximation to others that appear in the literature.

\section{Averages of rotationally symmetric functions on $S^2$} \label{sec:Theorem}

The proof of Lemma \ref{Lem-1} relies on rotational symmetries of the spherical harmonics; however some ambiguities can arise when discussing rotations. For this reason, we first lay out definitions that will be used in the proof.  

In $\mathbb R^3$, any orientation can be achieved by composing three elemental rotations, starting from a known standard direction. Let the standard direction be $(x,y,z)$ and the elemental matrices be
\[R_1(\cdot)=\left[\begin{array}{ccc}
\cos(\cdot) & 0 & \sin(\cdot)\\
0 & 1 & 0\\
-\sin(\cdot) & 0 & \cos(\cdot)\\
\end{array}\right]\quad\text{and}\quad R_2(\cdot)=\left[\begin{array}{ccc}
\cos(\cdot) & -\sin(\cdot) & 0 \\
\sin(\cdot) & \cos(\cdot) & 0\\
0 & 0 & 1\\
\end{array}\right].\]
The rotation $R_1(\cdot)$ rotates the $(y,z)$-plane around the $x$-axis using the right hand rule while $R_2(\cdot)$ rotates the $(x,y)$-plane around the $z$-axis using the right hand rule. 
The rotation matrix $R{(\rho,\beta,\alpha)}$ defined by
\[R(\rho,\beta,\alpha)=R_2(\rho)R_{1}(\beta)R_{2}(\alpha)\]
is intended to operate by pre-multiplying the column vector $(x,y,z)$ and represents an active rotation.\footnote{The matrices act on the coordinates of vectors defined in the initial fixed reference frame and give, as a result, the coordinates of a rotated vector defined in the same reference frame.}  Each matrix is meant to represent the composition of intrinsic rotations.\footnote{Rotations around the axes of the rotated reference frame.}  In terms of orbital parameters, if $\alpha=0$, then $\beta$ is the obliquity angle and $\rho$ is the precession angle.

Any point $(\hat x,\hat y, \hat z)$ can be decomposed in the elemental rotations relative to the standard direction as
\[\left[\begin{array}{c}
\hat x \\
\hat y\\
\hat z\\
\end{array}\right]=R(\rho,\beta,\alpha)\left[\begin{array}{c}
x\\
y\\
z\\
\end{array}\right].\]
Furthermore, write each set of coordinates in spherical coordinates as
\[\begin{array}{l}
x=\cos\theta\sin\phi\\
y=\sin\theta\sin\phi\\
z=\cos\phi
\end{array}\quad \text{and}\quad
\begin{array}{l}
\hat x=\cos\hat\theta\sin\hat\phi\\
\hat y=\sin\hat\theta\sin\hat\phi\\
\hat z=\cos\hat\phi
\end{array}\]
where $\theta$ and $\hat\theta$ are the azimuth angles as measured counterclockwise from the $x$- and $\hat x$-axes, respectively and $\phi$ and $\hat\phi$ are the usual polar angles measured relative to the positive vertical axis.  Notice that in planetary nomenclature, these angles give the co-latitudes (the angle measured relative to the equator is the latitude). 
Let  $R_{\rho,\beta,\alpha}$ denote the same rotation described by $R(\rho,\beta,\alpha)$ but which relates $(\hat\theta,\hat\phi)$ to $(\theta,\phi)$ so that
\[(\hat \theta,\hat \phi)=R_{\rho,\beta,\alpha}(\theta,\phi), \quad\text{and} \quad(\theta,\phi)=R_{\rho,\beta,\alpha}^{-1}(\hat\theta,\hat\phi).\]
Note that
\begin{align}
\cos\hat\phi=\hat z=\hat z(\alpha,\beta,\theta,\phi)=\cos\beta\cos\phi-\sin\beta\sin\phi\cos(\alpha+\theta).
\label{eq:z-hat}
\end{align}

\begin{lem} \label{Lem-1}
Suppose $f(\theta,\phi)$ is a square integrable function on $S^{2}\subset \mathbb{R}^3$  where $\theta$ is the azimuthal angle and $\phi$ is the polar angle.  Suppose also that there exists a coordinate system $(\hat\theta,\hat\phi)$ where the function $f$ does not depend on the azimulthal angle $\hat\theta$, i.e. there exists angles $\alpha$, $\beta,$ and $\rho$ such that the proper Euler rotation $R_{\rho,\beta,\alpha}$ yields a coordinate system
\[(\hat\theta,\hat\phi)=R_{\rho,\beta,\alpha}(\theta,\phi)\]
with $g(\hat\theta,\hat\phi)=f(R_{\rho,\beta,\alpha}^{-1}(\hat\theta,\hat\phi))$ and $\partial_{\hat\theta}g(\hat\theta,\hat\phi)=0$.
Then we have 
\[g(\cos\hat\phi)=g(\alpha,\beta,\theta,\phi)=\sum_{n=0}^\infty \tilde c_n \left(\sum_{k=-n}^{n} Y_{n}^k(\alpha,\beta) Y_{n}^{k}(\theta,\phi)\right)\]
where $Y_{n}^{k}(\theta,\phi)=N_{n,k}e^{ik\theta}P_{n}^k(\cos\phi)$ is the $n-k$ spherical harmonic with normalizing factor 
\[N_{n,k}=(-1)^k\sqrt{\frac{2n+1}{4\pi}\frac{(n-k)!}{(n+k)!}},\]
 associated Legendre polynomial $P_{n}^k$, and
\[\tilde c_n={2\pi}\int_{-1}^1g(\hat z)P_n(\hat z)d\hat z\]
where $P_n$ is the $n$-th Legendre polynomial.
Furthermore, for $\beta$ and $\theta$ fixed we have
\[\int_{0}^{2\pi}g(\alpha,\beta,\theta,\phi)d\theta=\int_{0}^{2\pi}g(\alpha,\beta,\theta,\phi)d\alpha=\sum_{n=0}^{\infty}\tilde b_nP_{n}(\cos\beta)P_{n}(\cos\phi)\]
where
\[\tilde b_n =\pi(2n+1)\int_{-1}^1 g(\hat z)P_n(\hat z)d\hat z.\]
\end{lem}

\begin{proof}
Notice that $g(\cos\hat\phi)=g(\hat z)$.  As stated above, the Legendre polynomials \(P_i(x), i = 0,1,2,\ldots ,\) form a complete orthogonal set in the space \(L^2([-1,1])\) with the properties \(P_i\) has degree \(i\) and \(P_i(1)=1\) \cite{BatemanII}. 

Expanding $g(\hat z)$ into its Legendre series gives
\begin{align}
g(\hat z)=\sum_{n=0}^\infty c_{n}P_{n}(\hat z),
\label{eq:f-legendre}
\end{align}
where
\[c_{n}=\frac{\int_{-1}^1 g (\hat z)P_{n}(\hat z)d\hat z}{\int_{-1}^1P_{n}(\hat z)^2d\hat z}=\frac{2n+1}{2}\int_{-1}^1 g (\hat z)P_{n}(\hat z)d\hat z\]
and $P_{n}$ is the $n$-th Legendre polynomial.  The series naturally converges in \(L^2\) and we will interpret the equal sign in Equation~(\ref{eq:f-legendre}) as equality in \(L^2\).

Changing back to spherical coordinates yields
\begin{align}
g(\cos\hat\phi)=\sum_{n=0}^\infty c_{n}P_{n}(\cos\hat\phi).
\label{eq:f-spherical-coord-legendre}
\end{align}

The addition formula for spherical harmonics \cite{Edmonds1974,Whittaker1990} says that
\[P_n(\cos\omega)=\frac{4\pi}{2n+1}\sum_{k=-n}^n Y_{n}^k(\theta,\phi)Y_{n}^k(\theta',\phi')^*\]
where 
\begin{align}
\cos\omega=\cos\phi\cos\phi'+\sin\phi\sin\phi'\cos(\theta-\theta')
\label{eq:cos-omega}
\end{align}
and $Y_{n}^k(\theta,\phi)=N_{n,k}e^{ik\theta}P_{n}^k(\cos\phi)$ is the $n-k$ spherical harmonic. 
Recall that 
\[\cos\hat\phi=\cos\beta\cos\phi-\sin\beta\sin\phi\cos(\alpha+\theta)\]
which can be written in the form of Equation \eqref{eq:cos-omega} by letting $\beta=-\tilde\beta$ and $\alpha=-\tilde\alpha$.
Then for any $n$ 
\begin{align*}
P_{n}(\cos\hat\phi)&=\frac{4\pi}{2n+1}\sum_{k=-n}^n Y_{n}^k(\theta,\phi)Y_{n}^k(\tilde\alpha,\tilde\beta)^*\\
&=\frac{4\pi}{2n+1}\sum_{k=-n}^n Y_{n}^k(\theta,\phi)Y_{n}^k(\alpha,\beta)
\end{align*}
because $Y_{n}^k(\alpha,\cdot)=Y_{n}^k(-\alpha,\cdot)^*$ and $Y_{n}^k$ is even in the second argument. Substituting the above into Equation \eqref{eq:f-spherical-coord-legendre} yields
\[g(\cos\hat\phi)=\sum_{n=0}^\infty \frac{4\pi c_n}{2n+1} \left(\sum_{k=-n}^{n} Y_{n}^k(\alpha,\beta) Y_{n}^{k}(\theta,\phi)\right).\]
Writing $g(\cos\hat\phi)=g(\cos\beta\cos\phi-\sin\beta\sin\phi\cos(\alpha+\theta))=g(\alpha,\beta,\theta,\phi)$ gives the formula from the statement of the theorem.

To prove that
\[\int_{0}^{2\pi}g(\alpha,\beta,\theta,\phi)d\theta=\int_{0}^{2\pi}g(\alpha,\beta,\theta,\phi)d\alpha=2\pi\sum_{n=0}^{\infty}P_{n}(\cos\beta)P_{n}(\cos\phi)\]
notice that 
\begin{align*}
\int_{0}^{2\pi}\sum_{n=0}^\infty \tilde c_n\sum_{k=-n}^n Y_{n}^k(\theta,\phi)Y_{n}^k(\alpha,\beta)d\theta&=\sum_{n=0}^{\infty} \tilde c_n\int_{0}^{2\pi}\sum_{k=-n}^n Y_{n}^k(\theta,\phi)Y_{n}^k(\alpha,\beta)d\theta
\end{align*}
because the function is absolutely integrable over a finite interval.  We see that
\begin{align*}
\int_{0}^{2\pi}\sum_{k=-n}^n Y_{n}^k(\theta,\phi)Y_{n}^k(\alpha,\beta)d\theta&=\sum_{k=-n}^n Y_{n}^k(\alpha,\beta)N_{n,k}P_{n}^k(\cos\phi)\int_{0}^{2\pi}e^{ik\theta}d\theta\\
&=\sum_{k=-n}^n Y_{n}^k(\alpha,\beta)N_{n,k}P_{n}^k(\cos\phi)(2\pi\delta_{k,0})
\end{align*}
where $\delta_{k,0}$ is the Kronecker Delta function indicating that the integral is zero except when $k=0$.  Then
\begin {align*}
\int_{0}^{2\pi}g(\alpha,\beta,\theta,\phi)d\theta&=2\pi\sum_{n=0}^{\infty}  \tilde c_n \left(N_{n,0}\right)^2P_{n}(\cos\beta)P_{n}(\cos\phi)\\
&=\sum_{n=0}^{\infty} \left(\frac{2n+1}{2}\right)\tilde c_nP_{n}(\cos\beta)P_{n}(\cos\phi)\\
&=\sum_{n=0}^{\infty} \tilde b_nP_{n}(\cos\beta)P_{n}(\cos\phi)
\end{align*}
Integrating in $\alpha$ yields the same result.
\end{proof}

\section{Integral of $x^{2k}\sqrt{(1-x^2)}$}  

The second lemma is instrumental for computing the coefficients of the Legendre series for the insolation distribution function.

\begin{lem}  \label{Lem-2}
For any non-negative integer \(k\),
\[
  \int_{-1}^1\sqrt{1-x^2}\,x^{2k}\,dx =
    (-1)^k\pi\choose{1/2}{k+1}.
\]
\end{lem}

\begin{figure}
\begin{center}
\includegraphics[width=.75\textwidth]{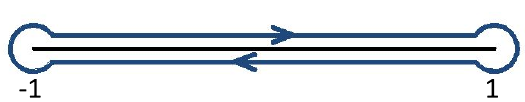}
\end{center}
\caption{The contour, $C$, used in the calculation of the integral in Lemma \ref{Lem-2}.}
\label{FIG-contour}
\end{figure}

\begin{proof}
Let
\[a_{2k}=\int_{-1}^1\sqrt{1-x^2}x^{2k}dx.\]
We compute $a_{2k}$ via the integral around the contour $C$ shown in Figure \ref{FIG-contour}.  The integral around $C$ is given by the residue at infinity of the integrand.  Namely
\begin{align*}
2a_{2k}&=2\int_{-1}^1\sqrt{1-x^2}x^{2k}dx\\
&=\int_{C}\sqrt{1-{z^2}}{z^{2k}}dz\\
&=-2\pi i\text{Res}( \sqrt{1-z^2}z^{2k},\infty)\\
&=2\pi i\text{Res}\left( \sqrt{1-\frac{1}{z^2}}\frac{1}{z^{2k}}\frac{1}{z^2},0\right)\\
&=-2\pi \text{Res}\left( \frac{\sqrt{1-{z^2}}}{z^{2k+3}},0\right)
\end{align*}
The series expansion of $\sqrt{1-z^2}$ is given by
\[\sqrt{1-z^2}=\sum_{n=0}^\infty \choose{1/2}{ n} (-z^2)^n=\sum_{n=0}^\infty (-1)^n \choose{1/2}{ n}z^{2n}.\]
were $\choose{1/2}{n}$ is a standard generalized binomial coefficient.  Multiplying through by $1/(z^{2k+3})$ yields
\[\frac{\sqrt{1-z^2}}{z^{2k+3}}=\sum_{n=0}^\infty (-1)^n \choose{1/2}{ n}z^{2n-2k-3}.\]
Then we calculate the residue as
\[\text{Res}\left( \frac{\sqrt{1-{z^2}}}{z^{2k+3}},0\right)=(-1)^{k+1}\choose{1/2}{k+1}\]
which establishes the formula given in the lemma.
\end{proof}

Alternatively, one could prove the above result using the Beta function of Euler \cite{Davis1972}
\begin{align}
B(a,b)=\int_{0}^1t^{a-1}(1-t)^{b-1}dt=\frac{\Gamma(a)\Gamma(b)}{\Gamma(a+b)},
\end{align}
with $a=k+1/2$ and $b=3/2$, as was noted by an anonymous reviewer. Using the Beta function results in a similarly succinct proof.

\section{Proof of Theorem \ref{Thm-Main}} \label{Proof-Main-Thm}
We begin with some necessary notation.  As was shown in McGehee and Lehman \cite{McGehee2012}, the ecliptic coordinations on the unit sphere are $(\hat x,\hat y, \hat z)\in S^2\subset\mathbb{R}^3$ where the $\hat z$-axis is perpendicular to the plane of the ecliptic, and the $\hat x$-axis ofter agrees with the major axis of the orbital ellipse.  In the planet-centric coordinates  $(x,y,z)\in S^2\subset\mathbb{R}^3$, the $z$-axis is the axis of planetary rotation and the $x$-axis follows the precession angle.  We note that these coordinates assume zero obliquity and zero precession for the planet.  In order to account for nonzero obliquity or precession, we rotate these coordinates first by the obliquity angle, $\beta$, then the precession angle, $\rho$, given by the rotation matrices $R_1(\beta)$ and $R_2(\rho)$ from Section \ref{sec:Theorem}.  If $\omega=2\pi/P$ where $P$ is the period of the planet's rotation about its axis then the rotation matrix multiplications
\[R(\rho,\beta,\omega t)=R_2(\rho)R_1(\beta)R_2(\omega t)
\]
account for the rotation of the planet, obliquity, and procession angles.  We can then relate the ecliptic coordinates to the planet-centric coordinates, namely
\[\left[\begin{array}{c}
\hat x \\
\hat y\\
\hat z\\
\end{array}\right]=R(\rho,\beta,\omega t)\left[\begin{array}{c}
x\\
y\\
z\\
\end{array}\right].\]
 We are now ready to prove the theorem.

\begin{proof} (Theorem \ref{Thm-Main})
McGehee and Lehman \cite{McGehee2012} showed that the annual average insolation at a latitude-longitude point $(\overline\phi,\theta)$ on the surface of a non-rotating planet with obliquity $\beta$ is proportional to
\begin{align}
\sqrt{1-(\sin\beta\cos\overline\phi\cos\theta-\cos\beta\sin\overline\phi)^2}
\label{EQ-insol-non-rotating}
\end{align}
where $\beta$ is the planet's obliquity.  We would like to include the rotation rate as well.  Following McGehee and Lehman's derivation, we see that the annual average insolation at a latitude-longitude point $(\overline\phi,\theta)$ at a particular time of day $t$ is given by
\begin{align}
I(\phi,\theta)=\sqrt{1-(\sin\beta\sin\phi\cos(\theta+\omega t)-\cos\beta\cos\phi)^2}
\label{EQ-insol-rotating}
\end{align}
where $\phi$ is the co-latitude (because cosine of the latitude is sine of the co-latitude).  Recalling equation \eqref{eq:z-hat}, the quantity in equation \eqref{EQ-insol-rotating} is
 proportional to the sine of the co-latitude in the ecliptic coordinates.  Normalizing so that the total insolation is $2$ gives us the insolation distribution
\[I(\hat\phi)=\frac{2}{\pi^2}\sin\hat\phi=\frac{2}{\pi^2}\sqrt{1-\hat z^2}\]
where $\hat\phi$ is the polar angle measured with respect to a vector perpendicular to the ecliptic plane.  The factor of two maintains compatibility with the usual normalization of insolation over one hemisphere.   
Since \eqref{EQ-insol-rotating} is square integrable,  Lemma \ref{Lem-1} applies.  Note that in our preliminary notation, we saw that we must take $\alpha=\omega t$ to account for the rotation of the planet.  Then
\begin{align}
I(\hat\phi)=I(\omega t,\beta,\theta,\phi)=\sum_{n=0}^\infty \tilde c_n\left(\sum_{k=-n}^n Y_{n}^k(\omega t,\beta)Y_n ^k(\theta,\phi)\right).
\label{EQ-insol-to-integrate}
\end{align}

The annual average insolation for a rapidly rotating planet may be gotten by integrating \eqref{EQ-insol-to-integrate} over the longitude $\theta$ or time $t$.  
Application of Lemma \ref{Lem-1} in either case yields
\[s(\phi,\beta)=\sum_{n=0}^\infty A_n P_n(\cos\beta) P_n(\cos\phi)\]
where
\begin{align*}
A_n &=\pi(2n+1)\int_{-1}^1\frac{2}{\pi^2}\sqrt{1-z^2}P_n(z)dz\\ 
&=\frac{2(2n+1)}{\pi}\int_{-1}^1\sqrt{1-z^2}P_{n}(z)dz
\end{align*}
If $n$ is odd, the integral is zero. Furthermore,
the coefficients of the Legendre polynomials are well-known \cite{BatemanII}.
In particular,
\[
  P_{2n}(x) = \sum_{k=0}^n p_{2n,2k}x^{2k},
\]
where
\[
  p_{2n,2k} =
    \frac{(-1)^{n-k}}{4^n} \choose{2n}{n-k} \choose{2n+2k}{2n} .
\]
Therefore, we can write
\begin{align*}
A_{2n} &=\frac{2(4n+1)}{\pi}\int_{-1}^1\sqrt{1-z^2}P_{2n}(z)dz\\ 
&=\frac{2(4n+1)}{\pi}\int_{-1}^1\sqrt{1-z^2}\left(\sum_{k=0}^n p_{2n,2k}z^{2k}\right)dz\\ 
&=\frac{2(4n+1)}{\pi}\sum_{k=0}^n p_{2n,2k} \int_{-1}^1\sqrt{1-z^2}z^{2k}dz\\
&=\frac{2(4n+1)}{\pi}\sum_{k=0}^n  \frac{(-1)^{n-k}}{4^n} \choose{2n}{n-k} \choose{2n+2k}{2n}  \int_{-1}^1\sqrt{1-z^2}z^{2k}dz
\end{align*}
Applying Lemma \ref{Lem-2} gives
\begin{align*}
A_{2n}&=\frac{(4n+1)(-1)^{n}}{2^{2n-1}} \sum_{k=0}^n  \choose{2n}{n-k} \choose{2n+2k}{2n} \choose{1/2}{k+1}.
\end{align*}

Finally, letting  $\eta=\sin\overline\phi=\cos\phi$ be the sine of the latitude (or cosine of the co-latitude) and $\zeta=\cos\beta$ be cosine of the obliquity we have
\begin{align*}
  s(\eta,\zeta) = \sum_{n=0}^\infty A_{2n}P_{2n}(\zeta)P_{2n}(\eta) ,
\end{align*}
which proves the formula.
\end{proof}

\section{Convergence of the Partial Sums}
\label{Section-Conv}

For a fixed obliquity, we can compute the convergence of the approximation in $L^2([-1,1])$  as well as $\mathcal C^0([-1,1])$. Convergence is slow in both norms and can also be quite slow pointwise at the poles for small obliquity.  This slow convergence is a contributing factor in our recommendation in Section \ref{Section-Comp} to use the Legendre approximation only up to the sixth degree in modeling scenarios.

Although convergence in $L^2$ is the natural space to consider from an analysis perspective, convergence in $\mathcal C^0$ may be more appropriate from a modeling perspective.  For example, because one needs to approximate $s(\eta,\zeta)$ in the Budyko-Widiasih energy balance model and modelers using that model are concerned with the qualitative behavior of solutions, convergence in the $\mathcal C^0$ norm may better capture qualitative differences between the approximations of different degrees.

\begin{figure}
\begin{center}
\includegraphics[width=.7\textwidth]{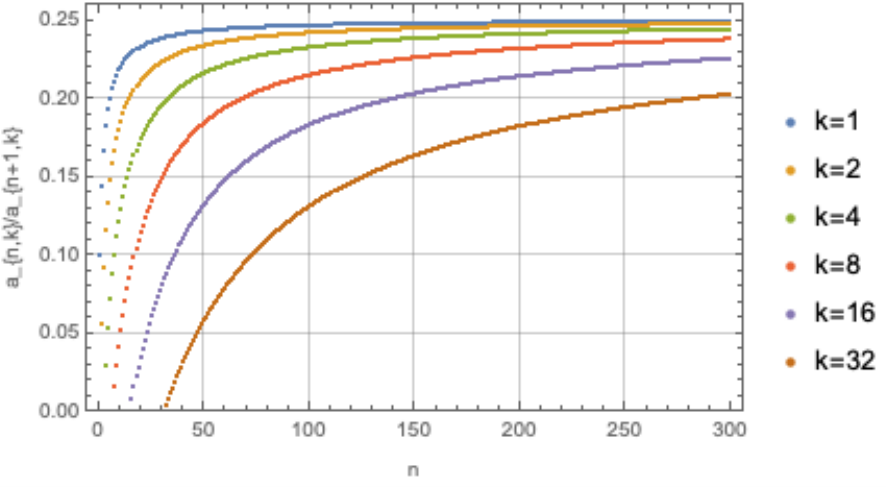}
\caption{The ratio $a_{n,k}/a_{n+1,k}$ for several $k's$ and $n<300$.  This shows that as $n$ gets large $a_{n+1,k}\approx4a_{n,k}$ and that when $k$ and $n$ are of comparable size, $a_{n+1,k}\gg a_{n,k}$. }
\label{FIG-ank}
\end{center}
\end{figure}

Below we present bounds on the convergence rate of the Legendre approximations.  We note that for $n\geq8$, great care must be taken in computing the coefficients $A_{2n}$ because while $A_{2n}\to0$ approximately like $-(2n)^{-2},$ the alternating terms that comprise $A_{2n}$, $a_{n,k}=\choose{2n}{n-k} \choose{2n+2k}{2n} \choose{1/2}{k+1}$, grow approximately like $a_{n+1,k}=4a_{n,k}$ for large $n$. We show the ratio of $a_{n,k}$ and $a_{n+1,k}$ in Figure \ref{FIG-ank}.  In this figure we see that for smaller $n$, $a_{n+1,k}$ is much larger than $4a_{n,k}$. The growth of $a_{n,k}$ and the resulting difficulty in computation is another major factor in the recommendation to use the degree six approximation in Section \ref{Section-Comp}.

Let $\sigma_{2N}(\eta,\zeta)$ denote the partial sum which is degree $2N$, i.e.
\begin{align}
\sigma_{2N}(\eta,\zeta)=\sum_{n=0}^{N}A_{2n}P_{2n}(\zeta)P_{2n}(\eta).
\end{align}
In the following subsections we demonstrate that convergence of these partial sums to the insolation function for $\zeta\not=\pm1$ is no slower than  $N^{-1}$ in $L^2$ and $N^{-1/2}$ in $\mathcal C^0$.  When $\zeta=\pm1$, the convergence is no slower than $N^{-1/2}$ in both $L^2$ and $\mathcal C^0$.

\subsection{Convergence in $L^2([-1,1])$}
\label{subsection:L2-convergence}

For convergence in $L^2([-1,1])$, we separate our argument into two parts, the first where $\zeta\not=\pm1$ and the second where $\zeta=\pm1$.  This separation is due to the fact that $s(\eta,\pm1)$ is not differentiable at the poles and, as a result, arguments and convergence rates differ.

\begin{prop}
For large $N$ and $\zeta\not=\pm1$, there exist constants $C_\zeta$ and $K_\zeta$ so that the $L^2$ difference of the approximation $\sigma_{2N}(\eta,\zeta)$ to the true distribution $s(\eta,\zeta)$ is bounded by
\begin{align}
\frac{C_\zeta}{N^{7/2}}\leq \|s(\eta,\zeta)-\sigma_{2N}(\eta,\zeta)\|_{L^2([-1,1])}\leq \frac{K_\zeta}{N}.
\end{align}
\end{prop}

\begin{proof}
It is routine to show that the convergence of the Legendre series approximations is at least order $N^{-k}$ when the $k$th derivative $\partial^k s(\eta,\zeta)/\partial\eta^k$ is in $L^2$, i.e. that for fixed $k$ and $\zeta$, the $L^2$ error is approximately
\begin{align}
\|s(\eta,\zeta)-\sigma_{2N}(\eta,\zeta)\|_{L^2([-1,1])}\approx\frac{C_\zeta}{(2N)^{-k}}\left\|\frac{\partial^k s(\eta,\zeta)}{\partial\eta^k}\right\|_{L^2([-1,1])}
\end{align}
for large $N$.  For $\zeta\not=\pm 1$, the derivative $\partial s(\eta,\zeta)/\partial\eta$ is bounded and thus in $L^2([-1,1])$, meaning convergence is at least $N^{-1}$.  

For $\zeta\not=\pm1$ the asymptotic behavior of the Legendre polynomials is given  by Szeg\"o \cite{Szego1975} (Theorem 8.21.2) as
\begin{align}
P_{k}(\cos\theta)=\frac{2}{\sqrt{2\pi k\sin\theta}}\cos\biggl(\left(k+\frac{1}{2}\right)\theta-\frac{\pi}{4}\biggr)+\mathcal O(k^{-3/2}).
\end{align}
In Figure \ref{FIG-P2n}, we see that for some values of $\zeta$, the second term in the asymptotic expansion of $P_{2n}(\zeta)$ cannot be neglected.  Bounding the Legendre polynomials as $C_\zeta(2n)^{-3/2}\leq|P_{2n}(\zeta)|$ for some constant $C_\zeta$ yields the estimate 
\begin{align}
2\sum_{n=N+1}^{\infty}\frac{|A_{2n}P_{2n}(\zeta)|^2}{4n+1}&\geq(C_\zeta)^2\sum_{n=N+1}^{\infty}\frac{1}{(2n)^8}
\end{align}
This suggests that the asymptotic behavior as $N\to\infty$ is constrained by the inequalities
\begin{align}
\|s(\eta,\zeta)-\sigma_{2N}(\eta,\zeta)\|_{L^2}\geq \frac{C_\zeta}{64(N+1)^{7/2}}
\end{align}
so that convergence is not more than $N^{-7/2}$.
\begin{figure}
\begin{center}
\includegraphics[width=.45\textwidth]{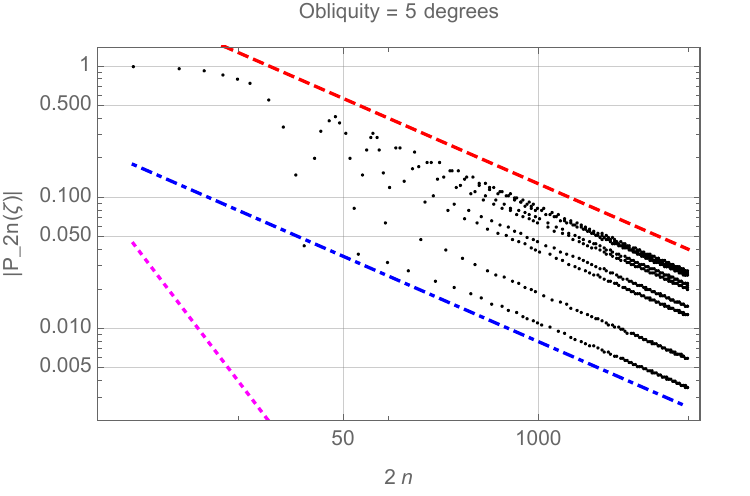}
\includegraphics[width=.45\textwidth]{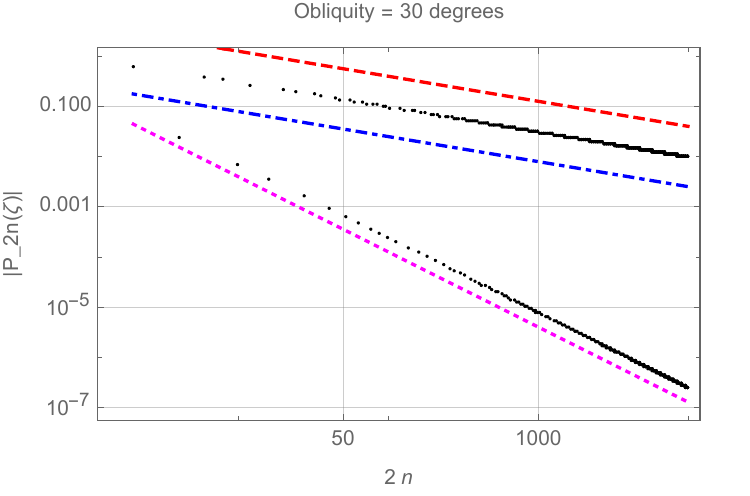}\\
\includegraphics[width=.45\textwidth]{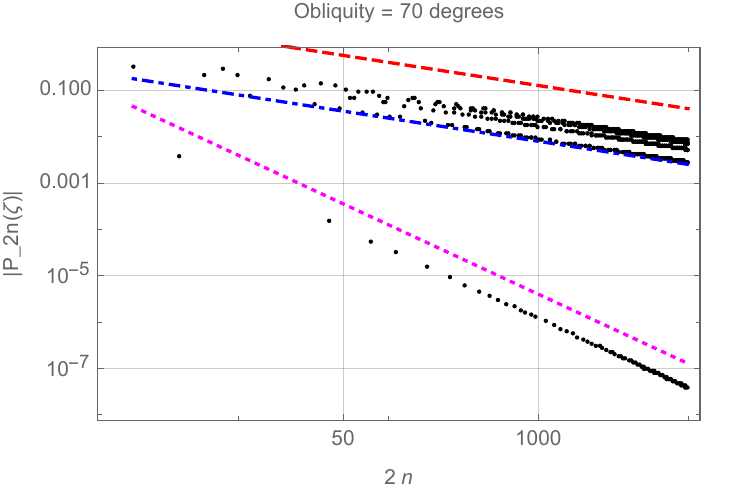}
\includegraphics[width=.45\textwidth]{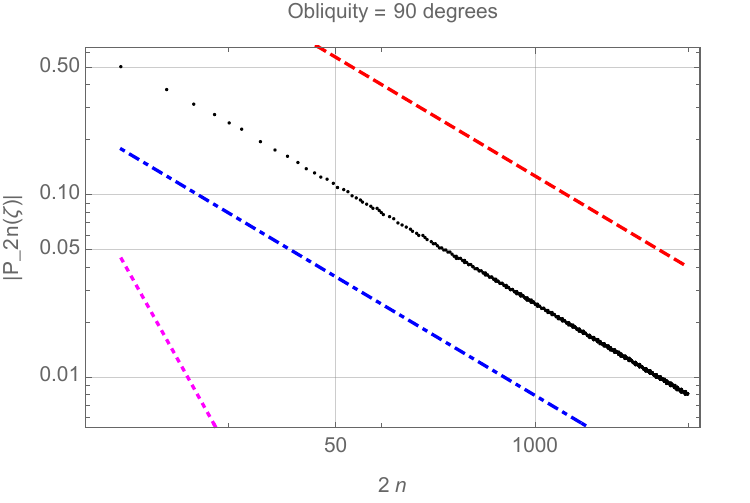}
\caption{The first 1000 values of the absolute value of the Legendre polynomials $|P_{2n}(\zeta)|$ for four values of cosine of obliquity (black dots) bounded by the curves $4(2n)^{-1/2}$ (red, dashed), $(2n)^{-1/2}/4$ (blue, dot-dashed), and $(2n)^{-3/2}/8$ (magenta, dotted). }
\label{FIG-P2n}
\end{center}
\end{figure}
\end{proof}

\begin{prop}
For large $N$ and $\zeta=\pm1$, there exist constants $C_\zeta$ and $K_\zeta$ so that the $L^2$ difference of the approximation $\sigma_{2N}(\eta,\pm1)$ to the true distribution $s(\eta,\pm1)$ is bounded by
\begin{align}
\frac{C_\zeta}{N^2}\leq \|s(\eta,\pm1)-\sigma_{2N}(\eta,\pm1)\|_{L^2([-1,1])}\leq \frac{K_\zeta}{\sqrt{N}}.
\end{align}
\end{prop}

\begin{proof}
For $\zeta=\pm 1$ the derivative with respect to sine of the latitude is not in $L^2$. In the Section~\ref{subsection:C0-convergence}, we show that
\begin{align}
|s(\eta,\zeta)-\sigma_{2N}(\eta,\zeta)|\leq\frac{1}{\sqrt{2N}}\left(\int_{-1}^1(1-\eta^2)\left(\frac{\partial s(\eta,\zeta)}{\partial\eta}\right)^2d\eta\right)^{1/2}.
\end{align}
Since the above integral is equal to 2 when $\zeta=\pm 1$, we see that convergence of the Legendre series in $L^2$ is no slower than $N^{-1/2}$ for large $N$.

To find an approximate lower bound for the convergence rate, we instead use the completeness of the Legendre polynomials in $L^2([-1,1])$ and the decay rate of $A_{2n}$ and $P_{2n}(\zeta)$ to estimate the error. Because the Legendre polynomials form a complete orthogonal set in $L^2$ we have a generalization of Parseval's identity
\begin{align}
\|s(\eta,\zeta)-\sigma_{2N}(\eta,\zeta)\|_{L^2}^2=\sum_{n=N+1}^{\infty}\frac{2}{4n+1}A_{2n}^2(P_{2n}(\zeta))^2.
\end{align}
For $n$ up to at least $2^{13}$  we have $2(2n)^{-2}\leq|A_{2n}|\leq4(2n)^{-2}$ (see Figure \ref{FIG-A2n} for a plot of the first 2500 $A_{2n}$).  Using the fact that $|P_{2n}(\pm 1)|=1$ and assuming that $2(2n)^{-2}\leq|A_{2n}|\leq4(2n)^{-2}$ holds for all $n$ we have
\begin{align}
 2\sum_{n=N+1}^{\infty}\frac{|A_{2n}P_{2n}(\pm1)|^2}{4n+1}\geq \frac{1}{2}\sum_{n=N+1}^{\infty}\frac{1}{(2n)^5}
\end{align}
which implies
\begin{align}
\|s(\eta,\pm1)-\sigma_{2N}(\eta,\pm1)\|_{L^2}\geq \frac{1}{8(N+1)^{2}}
\end{align}
so that the asymptotic convergence rate is no faster than $N^{-2}$.
\begin{figure}
\begin{center}
\includegraphics[width=.7\textwidth]{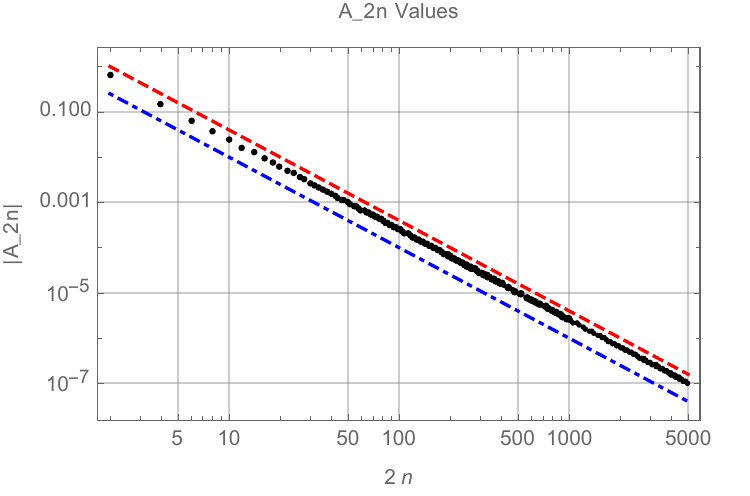}
\caption{The first 2500 values of $|A_{2n}|$ (black dots) bounded by the curves $4(2n)^{-2}$ (red, dashed) and $(2n)^{-2}$ (blue, dot-dashed).}
\label{FIG-A2n}
\end{center}
\end{figure}
\end{proof}

Numerical computations when $\zeta=\pm1$ suggest that the convergence rate is slightly less than 1 (see Table \ref{TAB-L2-0}).

\begin{table}
\caption{Approximate convergence rates in the $L^2$ norm for the approximation to the insolation on a planet with no axial tilt.}
\begin{center}
    \begin{tabular}{ | c | r | p{3cm} |}
    \hline
    $2N$ &$ \|s(\eta,1)-\sigma_{2N}(\eta,1)\|_{L^2}$ & Approx. order of convergence \\ \hline \hline
    2 & $7.27019\cdot 10^{-2}$ &  \\ \hline
    4 & $2.97384\cdot 10^{-2}$ & $1.28967$  \\ \hline
    8 & $1.02277\cdot 10^{-2}$ & $1.53984$ \\ \hline
    16 & $3.36645\cdot 10^{-3}$ & $1.60319 $  \\ \hline
    32 & $1.32593\cdot 10^{-3}$ & $1.34423$ \\ \hline
    64 & $6.32304\cdot 10^{-4}$ & $1.06831$ \\ \hline
    128 & $3.12467\cdot 10^{-4}$ & $1.01691 $  \\ \hline
    256 & $1.56444\cdot 10^{-4}$ & $9.98062 \cdot 10^{-1} $  \\ \hline
    512 & $7.84126\cdot 10^{-5}$ & $9.96485 \cdot 10^{-1}$  \\ \hline
    \end{tabular}
\end{center}
\label{TAB-L2-0}
\end{table}


\subsection{Convergence in $\mathcal C^0([-1,1])$}
\label{subsection:C0-convergence}

For modelers, the $L^2$ norm may not be the ideal measure for goodness-of-fit of the approximation.  Instead, the $\mathcal C^0$ norm may be more appropriate. Take the norm in $\mathcal C^0([-1,1])$ to be
\begin{align}
\|f\|_{\mathcal C^0([-1,1])}=\sup_{x\in[-1,1]}|f(x)|.
\end{align} 

\begin{prop}
For large $N$ and $\zeta$ fixed, there exists a constant $C_\zeta$ so that the $\mathcal C^0$ difference of the approximation $\sigma_{2N}(\eta,\pm1)$ to the true distribution $s(\eta,\pm1)$ is bounded by
\begin{align}
\|s(\eta,\pm1)-\sigma_{2N}(\eta,\pm1)\|_{\mathcal C^0([-1,1])}\leq \frac{C_\zeta}{\sqrt{N}}.
\end{align}
\end{prop}

As mentioned in the previous section we will bound the error of the approximation with 
\begin{align}
|s(\eta,\zeta)-\sigma_{2N}(\eta,\zeta)|\leq\frac{1}{\sqrt{2N}}\left(\int_{-1}^1(1-\eta^2)\left(\frac{\partial s(\eta,\zeta)}{\partial\eta}\right)^2d\eta\right)^{1/2}
\end{align}
with the standard use of integration by parts twice and application of the Cauchy-Schwarz inequality.
First write
\begin{align}
\sigma_{2N}(\eta,\zeta)&=\sum_{n=0}^N A_{2n}P_{2n}(\zeta)P_{2n}(\eta)=\int_{-1}^1\left[\sum_{n=0}^N\left(\frac{4n+1}{2}\right)P_{2n}(x)P_{2n}(\eta)\right]s(x,\zeta)dx.
\end{align}
Using
\begin{align}
-\frac{d}{dx}(1-x^2)\frac{d}{dx}P_n(x)=n(n+1)P_n(x)
\end{align}
and integration by parts twice, we can write
\begin{align}
\sigma_{2N}(\eta,\zeta)&=\int_{-1}^1\left[\sum_{n=0}^N\left(\frac{4n+1}{4n(2n+1)}\right)P_{2n}(x)P_{2n}(\eta)\right]\left(-\frac{d}{dx}(1-x^2)\frac{d}{dx}s(x,\zeta)\right)dx.
\end{align}
Similarly,
\begin{align}
s(\eta,\zeta)&=\int_{-1}^1\left[\sum_{n=0}^\infty\left(\frac{4n+1}{4n(2n+1)}\right)P_{2n}(x)P_{2n}(\eta)\right]\left(-\frac{d}{dx}(1-x^2)\frac{d}{dx}s(x,\zeta)\right)dx
\end{align}
so that
\begin{align}
|s(\eta,\zeta)-\sigma_{2N}(\eta,\zeta)|&=\left|\int_{-1}^1\left[\sum_{n=N}^\infty\left(\frac{4n+1}{4n(2n+1)}\right)P_{2n}(x)P_{2n}(\eta)\right]\left(-\frac{d}{dx}(1-x^2)\frac{d}{dx}s(x,\zeta)\right)dx\right|.
\end{align}
Integrating by parts to move a derivative to $P_{2n}(x)$ and applying the Cauchy-Schwarz inequality allows us to write
\begin{align}
|s(\eta,\zeta)-\sigma_{2N}(\eta,\zeta)|&\leq\left(\int_{-1}^1\left[\sum_{n=N}^\infty\left(\frac{4n+1}{4n(2n+1)}\right)P_{2n}(\eta)\frac{d}{dx}P_{2n}(x)\right]^2(1-x^2)dx\right)^{1/2}\\
&\qquad\qquad\cdot\left(\int_{-1}^1(1-x^2)\left[\frac{d}{dx}s(x,\zeta)\right]^2dx\right)^{1/2}.
\end{align}
The derivatives of the Legendre polynomials are orthogonal with the weight function $(1-x^2)$
\begin{align}
\int_{-1}^1(1-x^2)\left(\frac{d}{dx}P_m(x)\right)\left(\frac{d}{dx}P_n(x)\right)dx=\frac{2n}{2n+1}\delta_{mn}.
\end{align}
Using this relationship and using 1 for the bound of the Legendre polynomials squared allows us to write
\begin{align}
|s(\eta,\zeta)-\sigma_{2N}(\eta,\zeta)|&\leq\left(\sum_{n=N+1}^\infty\left(\frac{4n+1}{4n(2n+1)}\right)^2\frac{2n}{2n+1}\right)^{1/2}\cdot\left(\int_{-1}^1(1-x^2)\left[\frac{d}{dx}s(x,\zeta)\right]^2dx\right)^{1/2}\\
&\leq2\left(\sum_{n=N+1}^\infty\frac{1}{(2n)^2}\right)^{1/2}\cdot\left(\int_{-1}^1(1-x^2)\left[\frac{d}{dx}s(x,\zeta)\right]^2dx\right)^{1/2}\\
&\leq \frac{1}{\sqrt{2N}}\left(\int_{-1}^1(1-x^2)\left[\frac{d}{dx}s(x,\zeta)\right]^2dx\right)^{1/2}.
\end{align}
The integral is finite for any $\zeta\in[-1,1]$, so the convergence rate in $\mathcal C^0$ is no slower than $N^{-1/2}$.  

\begin{cor}
For fixed $\zeta$, the convergence of $\sigma_{2N}(\eta,\zeta)$ to $s(\eta,\zeta)$ is uniform.
\end{cor}

\begin{proof}
From the above argument, we have 
\begin{align}
|s(\eta,\zeta)-\sigma_{2N}(\eta,\zeta)|&\leq \frac{1}{\sqrt{2N}}\left(\int_{-1}^1(1-x^2)\left[\frac{d}{dx}s(x,\zeta)\right]^2dx\right)^{1/2}
\end{align}
and the integral is finite for any $\zeta\in[-1,1]$.
\end{proof}

\subsection{Approximations up to Degree Eight}

In the next section, we assert that the sixth degree Legendre approximation is the preferable degree approximation to choose when modeling annual average insolation on rapidly rotating planets.  Here we compare the first four Legendre series approximations but showing their qualitative differences as well as error in the $L^2$ and $\mathcal C^0$ norms as a function of obliquity.

The biggest difference between the approximations is how well they capture the qualitative aspects of the insolation distribution.  A second order approximation is sufficient for Earth's obliquity ($\beta\approx23.4^\circ$), however other obliquity angles produce qualitatively different insolation distributions and the second order approximation is no longer sufficient to capture an accurate insolation distribution.  In particular, for obliquities between $47^\circ$ and $63^\circ$ (and $117^\circ$ and $133^\circ$), the insolation has a characteristic `W' shape.  This `W' shape is first captured for all obliquities by the degree six approximation, see Figure \ref{FIG-Insol-2-4-6}.

\begin{figure}
\begin{center}
\includegraphics[width=.49\textwidth]{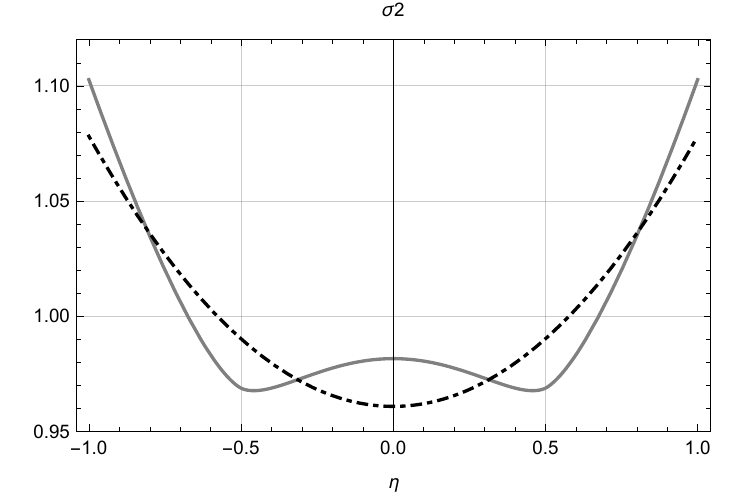}
\includegraphics[width=.49\textwidth]{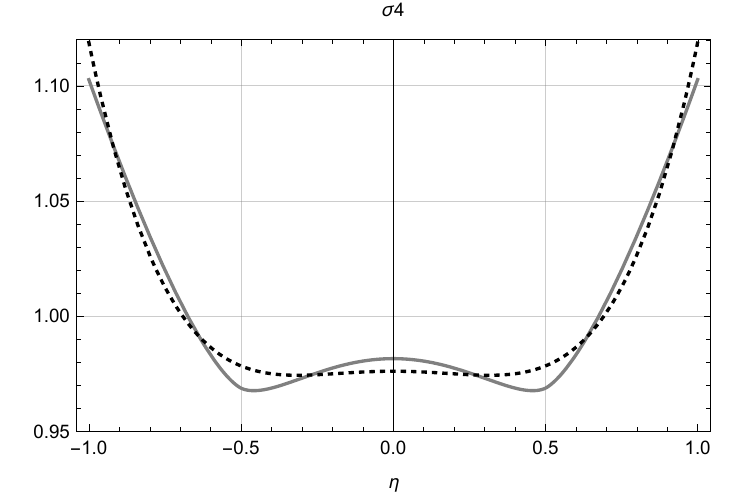}\\
\includegraphics[width=.49\textwidth]{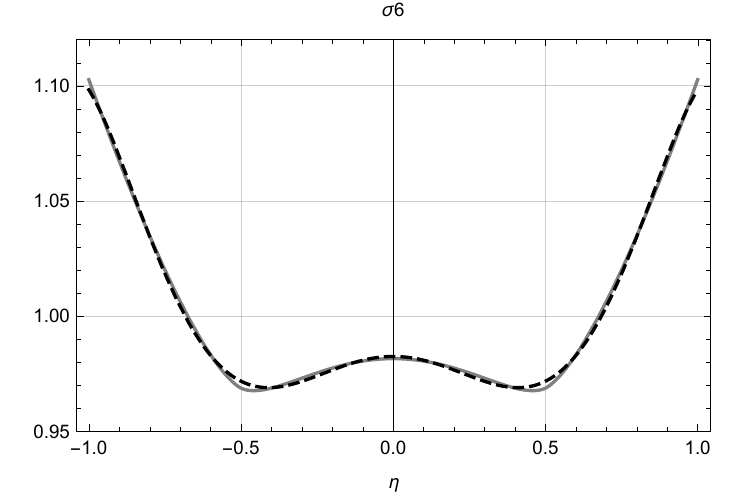}
\includegraphics[width=.49\textwidth]{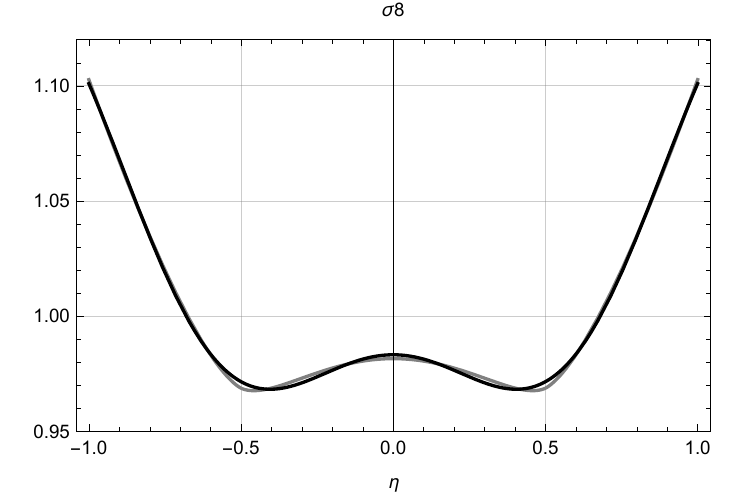}
\end{center}
\caption{The insolation distribution $s(\eta,\zeta)$ (gray) and the degree two (dot-dashed), degree four (dotted), degree six (dashed), and degree eight (solid) approximations for an obliquity angle of $120^\circ$ (approximately Pluto's obliquity). Notice that the degree six approximation is the lowest degree approximation that is able to capture the qualitative `W' shape of the distribution.}
\label{FIG-Insol-2-4-6}
\end{figure}

The $L^2$ error as a function of the obliquity is given in Figure \ref{FIG-L2-Norm-2-4-6}.  In this figure, we see that the $2N$ approximation is bounded above by the $2(N-1)$ approximation; however, it is interesting to note that at some obliquities, there is no reduction in the $L^2$ error between successive approximations and one must take an additional term in the Legendre approximation to ensure a decrease in the error.

\begin{figure}
\begin{center}
\includegraphics[width=.49\textwidth]{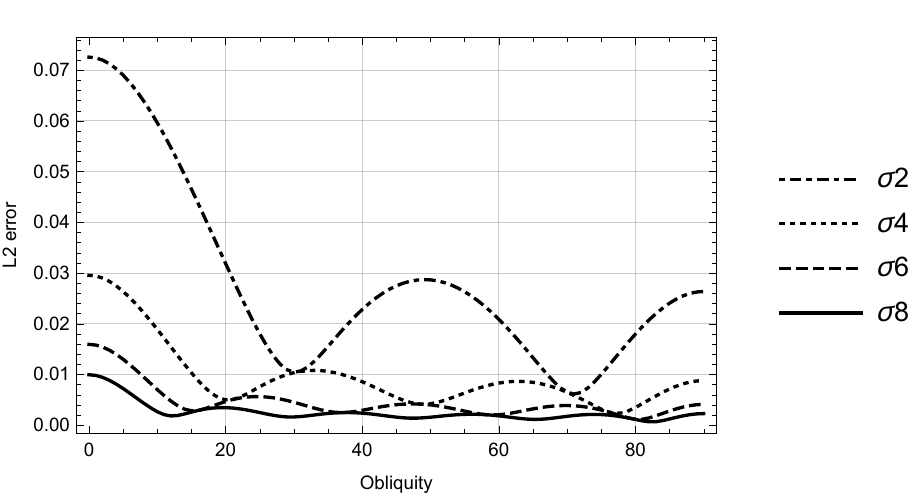}
\includegraphics[width=.49\textwidth]{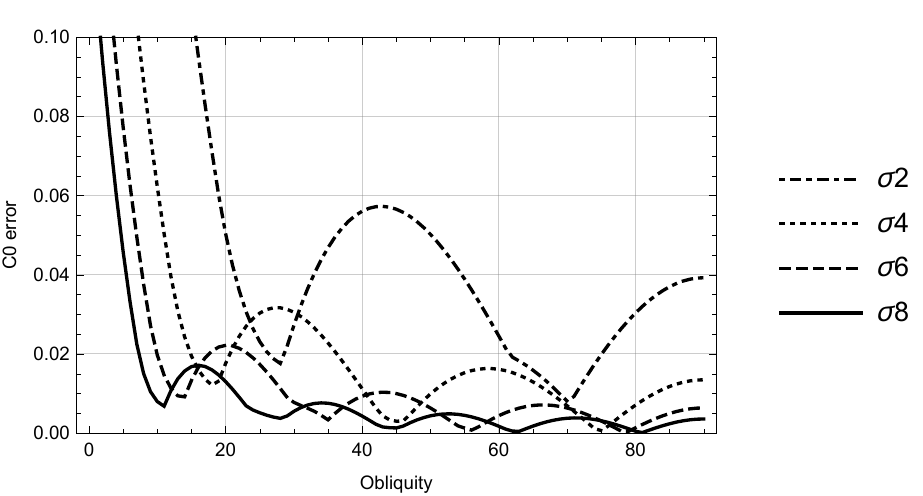}
\end{center}
\caption{The $L^2$ (left) and $\mathcal C^0$ (right) error for the degree two (dot-dashed), degree four (dotted), degree six (dashed), and degree eight (solid) approximations as a function of obliquity.  Note that these errors are symmetric about an obliquity angle of $90^\circ$. }
\label{FIG-L2-Norm-2-4-6}
\end{figure}

The $\mathcal C^0$ error as a function of the obliquity is given in Figure \ref{FIG-L2-Norm-2-4-6}.  Here the decrease in error is not monotonic for all obliquities as we increase the degree of the approximation.  For example, the degree two approximation is better than the degree four approximation for obliquity angles between $25^\circ$ and $32^\circ$.   As with the $L^2$ error, we see that one must take an additional term in the Legendre approximation to ensure a decrease in the error, although it is not clear if this is true for all $N$.


\section{Comparisons with Other Approximations}
\label{Section-Comp}

Finding the annual average insolation distribution for a planet is not a new problem, and various other approximations exist.  The sixth degree polynomial approximation, $\sigma_6$, gotten by truncating Equation \ref{EQ:Insol-Approx} at $n=3$ is given by
\begin{equation}
\begin{aligned}
\sigma_6(\eta,\beta)&=1-A_2P_2(\cos\beta)P_2(\eta)-A_4P_4(\cos\beta)P_4(\eta)-A_6P_6(\cos\beta)P_{6}(\eta)
\label{EQ-insolation}
\end{aligned}
\end{equation}
where $A_2=5/8$, $A_4=9/64$, and $A_6=65/1024$ and the $P_k$'s are the Legendre polynomials
\begin{align}
P_2(y)&=(3y^2-1)/2,\\
P_4(y)&=(35y^4-30y^2+3)/8,\\
P_6(y)&=(231y^6-315y^4+105y^2-5)/16.
\end{align}
The approximation $\sigma_6$ is preferable to other approximations in the literature because
\begin{enumerate}
\item it is the only approximation with explicit dependence on obliquity that is also continuous,
\item it is a better approximation in the $L^2$ norm than any other polynomial approximation of equal or lesser degree,
\item for obliquity angles between $7^\circ$ and $183^\circ$, $\sigma_6$ has smaller error in the $L^2$  norm than  any other approximation in the literature, 
\item for obliquity angles between $3^\circ$ and $23^\circ$ and $157^\circ$ and $187^\circ$, the error in the $\mathcal C^0$ norm for $\sigma_6$ is no worse than the same order of magnitude as the error for other approximations in the literature, 
\item for obliquity angles between $23^\circ$ and $157^\circ$, $\sigma_6$ is the best approximation in the $\mathcal C^0$ norm compared to  other approximations in the literature, 
\item $\sigma_6$ is the lowest degree polynomial that captures the qualitative distribution of insolation for all obliquities.
\end{enumerate}
The polynomial approximation $\sigma_6$ should be used instead of the integral form of the annual average insolation distribution function and other approximations in the literature because the approximation is more computationally efficient and sufficiently accurate to capture the qualitative characteristics of the actual distribution function for any obliquity.


\subsection{Approximations for Earth's Insolation Distribution}

North \cite{North1975a} explicitly gives a second degree approximation for the insolation distribution of Earth derived from orbital movements.  North gives his approximation in terms of a second degree Legendre approximation as
\[\sigma_\text{North}(\eta)=1-0.482P_2(\eta)\]
 (see equation 2 in \cite{North1975a}). The coefficient is gotten by setting $\zeta=\cos(23^\circ)$ in $A_2P_2(\zeta)$. As we saw in the previous section, the sixth degree approximation will give a better approximation in terms of error in the $L^2$ and $\mathcal C^0$ norms than this second degree Legendre approximation.
 
Chylek and Coakley \cite{Chylek1975} also compute Earth's mean annual insolation distribution from first principles.  They give their approximation as a function of sine of the latitude as a piecewise linear function broken up at intervals of $5^\circ$ in latitude.  The values that their approximation takes at these breaks are given in Table 1 of \cite{Chylek1975}. The sixth degree Legendre approximation does better than Chylek and Coakley's approximation in the $L^2$ norm but worse in the $\mathcal C^0$ norm.  It is worse in the $\mathcal C^0$ norm because of the poor approximation at the poles due to the slow convergence there.

\subsection{Approximation by Ojakangas and Stevenson}
 
The approximation given in Ojakangas and Stevenson \cite{Ojakangas1989} gives an explicit dependence on obliquity and is given by
\begin{align}
\sigma_\text{OS}(\eta,\zeta)=\begin{cases}
\frac{4\sqrt{1-\eta^2}}{\pi}&\eta>\zeta\\
\frac{4\sqrt{\arccos^2(\zeta)+\arccos^2(\eta)}}{\sqrt{2}\pi}& \eta\leq\zeta
\end{cases}
\end{align}
where $\eta$ is the sine of the latitude and $\zeta$ is cosine of the obliquity (see equations A.17 and A.18 in \cite{Ojakangas1989}).  Ojakangas and Stevenson developed this approximation to understand the insolation on Europa and note that their approximation $
\sigma_\text{OS}$ is valid only when $\zeta$ is close to 1 (i.e. the obliquity close to zero) \cite{Ojakangas1989}.  Ojakangas and Stevenson use $\arccos\zeta=3^\circ$.  The degree six approximation with $\zeta=\cos(3^\circ)$ is
\begin{align}
\sigma_6(\eta,\cos(3^\circ))\approx1-3.1\cdot 10^{-1}P_2(\eta)-5.13\cdot 10^{-2}P_4(\eta)-1.87\cdot 10^{-2}P_6(\eta).
\end{align}
We note that the Ojakangas and Stevenson approximation $\sigma_{OS}$ is continuous only when $\zeta=1$. For all other obliquities, the approximation $\sigma_{OS}$ is discontinuous at $\eta=\pm\zeta$ and the jump distance increases as the obliquity increases.

%
 
In Figure \ref{FIG-C0-Norms-Ojakangas}, we show the error of the approximations to the insolation distribution $s(\eta,\zeta)$ in the $L^\infty$ and $L^2$ norms as a function of the obliquity.  The Ojakangas and Stevenson approximation (solid line) is better than the sixth degree Legendre approximation (dashed line) in both the $L^\infty$ and $L^2$ norms for obliquities smaller than $7^\circ$.   When modeling planets with low obliquity angle, the Ojakangas and Stevenson approximation may be preferable over the degree six Legendre approximation provided the discontinuity at $\eta=\zeta$ is not an issue in the modeling framework.

\begin{figure}
\begin{center}
\includegraphics[width=.455\textwidth]{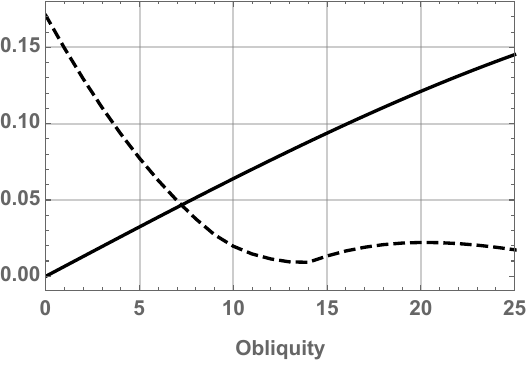}
\includegraphics[width=.45\textwidth]{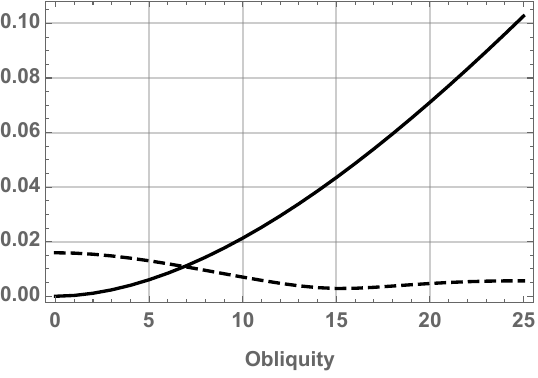}
\end{center}
\caption{The error of the approximations in the $L^\infty$ norm (left)  and the $L^2$ norm (right) as a function of the obliquity. Solid curves are the Ojaknagas and Stevenson approximation, $\sigma_{OS}(\eta,\zeta)$ and dashed curves are the sixth degree Legendre approximation, $\sigma_6(\eta,\zeta).$}
\label{FIG-C0-Norms-Ojakangas}
\end{figure}

\subsection{Sellers Equation for Insolation}

Sellers gives the daily insolation distribution as a function of latitude $\overline\phi$ and time since the northern hemisphere winter solstice is given by
\begin{align}
s_\text{Sellers}(\overline\phi,t)=\frac{\sin\overline\phi}{\pi}(\sin(\delta(t)H(\overline\phi,t))+\sin H(\overline\phi,t)\cos\delta(t))
\end{align}
where 
\begin{align}
\sin\delta(t)&=-\sin\beta\cos L(t)\\
\cos H(\overline\phi,t)&=-\cos^{-1}(\tan\overline\phi\tan\delta(t)), \ 0\leq H\leq \pi
\end{align}
and $\beta$ is the obliquity and $L$ is the longitude of the planet in its orbit measured relative to northern hemisphere's winter solstice \cite{Sellers1965}. Note that $\delta$ is the solar declination and $H$ gives the length of the half-day from sunrise to sunset. This equation for insolation has been used in models of Earth's climate \cite{Sellers1965,North1979} as well as Mars \cite{Nakamura2002}.

Rose et al.  find a Fourier-Legendre series approximation for Sellers' insolation function with Legendre polynomials for the latitudinal dependence and a Fourier cosine series for the time dependence \cite{Rose2017}.   Their approximation is
\begin{align}
\sigma_\text{RCB}(\eta,t)= 1+s_{11}(\beta)\cos(\omega t)P_1(\eta)+(s_{20}(\beta)+s_{22}(\beta)\cos(2\omega t))P_2(\eta)
\end{align}
where $\omega=2\pi/P$ and $P$ is the orbital period and 
\begin{align}
s_{11}(\beta)&=-2\sin\beta,\\
s_{20}(\beta)&=-\frac{5}{16}(2-3\sin^2\beta),\\
s_{22}(\beta)&=\frac{15}{16}\sin^2\beta,
\end{align}
Averaging this approximation over one orbital period yields the second degree Legendre approximation from this paper
\begin{align}
\frac{1}{P}\int_{0}^P\sigma_\text{RCB}(\eta,t)dt= 1+P_2(\cos(\beta))P_2(x)=\sigma_2(\eta,\cos\beta).
\end{align}
 
 North and Coakley also find a Fourier-Legendre series approximation \cite{North1979}, but  remove the dependence on obliquity and give only coefficients relevant for Earth's obliquity
\begin{align}
\sigma_\text{NC}(\eta,t)= 1+S_1\cos(2\pi t)P_1(\eta)+(S_2+S_{22}\cos(4\pi t))P_2(\eta)
\end{align}
where $S_1=-0.796$, $S_2=-4.77$, and $S_{22}=0.147$.  Again, averaging this approximation over one period of the orbit yields the Legendre series, which in this case is also North's approximation
\begin{align}
\frac{1}{P}\int_{0}^P\sigma_\text{NC}(\eta,t)dt= 1+S_2P_2(x)\approx\sigma_\text{North}(\eta).
\end{align}

\section{Discussion}

We have derived the Legendre series expansion for the insolation distribution on rapidly rotating planets as a function of sine of the latitude and the planet's obliquity. Furthermore, we give an explicit formula for the coefficients of this series as it depends on the obliquity.

\begin{figure}
\begin{center}
\includegraphics[width=.32\textwidth]{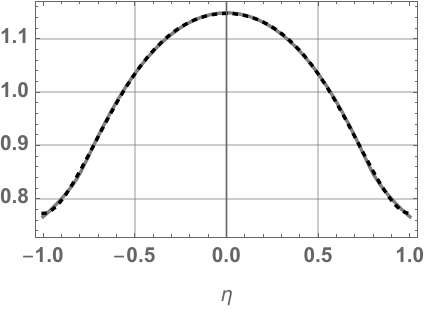}
\includegraphics[width=.32\textwidth]{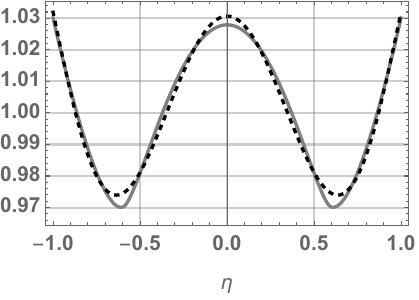}
\includegraphics[width=.32\textwidth]{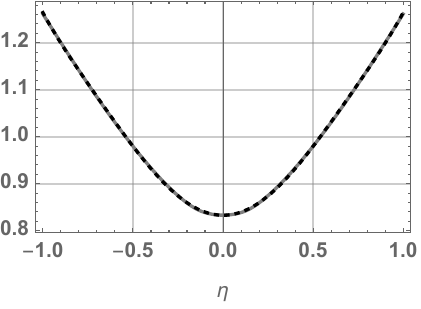}
\end{center}
\caption{Insolation distribution for various obliquity values in Mars' orbital history.  In all plots we give the actual (gray) and sixth order approximation to the insolation distribution (black dashed). Note: the vertical scales are different across plots and the actual and approximations plots are nearly identical by eye for the left and right figures.  In (a) we have $\beta=37^\circ$ which is Mars' average obliquity as given in Laskar, in (b) we have $\beta=54^\circ$ which gives an example where polar and equatorial regions receive about the same yearly insolation, and in (c) we have $\beta=82^\circ$, which is the maximal possible obliquity}
\label{Insolation-Mars}
\end{figure}

Being able to compute insolation by latitude as an explicit function of obliquity is particularly important when modeling exoplanets or in the case of Mars due to the chaotic nature of Mars' obliquity over the course of 5 billion years. Laskar et al., \cite{Laskar2004}, showed that the obliquity of Mars ranges from $\beta\approx24^\circ$ to $\beta\approx82^\circ$.  Over this range in obliquity, the insolation distribution changes drastically, going from a downward facing parabolic shape, to a strong `W' shape, to an upward facing parabolic shape (see Figure \ref{Insolation-Mars}).  In modeling the climate of Mars over time, it would be necessary to have an algebraic representation of the insolation with explicit dependence on obliquity. 

We also compare finite truncations of this series to other approximations which exist in the literature.  We conclude that the sixth degree approximation is the optimal approximation to use for rapidly rotating planets because it is continuous in the sine of the latitude, it has smaller error in both the $L^2$ and $\mathcal C^0$ norms compared to other approximations, it has explicit dependence on the obliquity, and it is the lowest degree approximation that captures the qualitative characteristics of the distribution for all obliquities.

\section*{Declarations of Interest}
None.

\section*{Funding}
This work was supported by the Mathematics and Climate Research Network (NSF Grants DMS-0940366 and DMS-0940363), the Mathematical Sciences Postdoctoral Research Fellowship  (Award Number DMS-1902887), and an Interdisciplinary Doctoral Fellowship from the Graduate School at the University of Minnesota. Funding sources did not have any role in study design; the collection, analysis and interpretation of data and methods; the writing of the report; and the decision to submit the article for publication.

\section*{Acknowledgements}
The authors would like to thank Nikole Lewis for help in framing the application and an anonymous reviewer for their close reading and constructive comments for improvement of the manuscript.





\end{document}